\documentclass[copyright,creativecommons]{eptcs}

\usepackage{amsmath, amssymb, latexsym}

\newtheorem{thm}{Theorem}

\newtheorem{theorem}[thm]{Theorem}
\newtheorem{corollary}[thm]{Corollary}
\newtheorem{lemma}[thm]{Lemma}
\newtheorem{proposition}[thm]{Proposition}

\newtheorem{definition}[thm]{Definition}

\newcommand{\cA}{\mathcal{A}}

\newcommand{\cL}{\mathcal{L}}
\newcommand{\cM}{\mathcal{M}}

\newcommand{\cO}{\mathcal{O}}

\newcommand{\cX}{\mathcal{X}}

\usepackage{color}
\newcommand{\myComment}[1]{}
\newcommand{\Model}{(X, \tau, \Phi, V)}

\newcommand{\br}[1]{[\![ #1]\!]}
\newcommand{\ds}[1]{d #1}
\newcommand{\po}[1]{<^S_d #1}
\newcommand{\intr}[1]{\mathit{int} #1 }
\newcommand{\Intr}[1]{\mathit{Int} #1 }

\newcommand{\intrh}[1]{\widehat{\intr{ #1}} }

\newcommand{\ayycolor}[1]{{\color{black}{#1}}}

\newcommand{\aycolor}[1]{{\color{black}{#1}}}				
			
\newcommand{\scomm}[1]{{\color{black}{#1}}}
\newcommand{\sse}{\subseteq}
\newcommand\pto{\rightharpoonup}
\newcommand{\aybuke}[1]{}

\renewcommand{\phi}{\varphi}
\newcommand{\weg}[1]{}

\newcommand{\II}[1]{\lbr #1 \rbr} 

\newcommand{\lbr}{[\![}
\newcommand{\rbr}{]\!]}

\newcommand{\imp}{\rightarrow}
\newcommand{\Imp}{\Rightarrow}

\newcommand{\et}{\wedge}
\newcommand{\vel}{\vee}

\newcommand{\Dia}{\Diamond}

\newcommand{\M}{\hat{K}}

\renewcommand{\phi}{\varphi}

\newcommand{\inter}{\cap}

\newcommand{\Boxx}{\Box}  

\usepackage{url}

\usepackage{enumitem}

\setenumerate[1]{itemsep=.2ex,partopsep=0pt,parsep=\parskip,topsep=.5ex}

\setitemize[1]{itemsep=.2ex,partopsep=0pt,parsep=\parskip,topsep=.8ex}

\newcommand{\Sub}{\mathit{Sub}}

 \def\pushright#1{{
    \parfillskip=0pt            
    \widowpenalty=10000         
    \displaywidowpenalty=10000  
    \finalhyphendemerits=0      
    \leavevmode                 
    \unskip                     
    \nobreak                    
    \hfil                       
    \penalty50                  
    \hskip.2em                  
    \null                       
    \hfill                      
    {#1}                        
    \par}}                      
 \def\qed{\pushright{\rule{2mm}{3mm}}\penalty-700 \smallskip}

\newenvironment{proof}[1]
{\noindent{\bf Proof} #1}%
{\qed}

\title{Announcement as effort on topological spaces}
\date{\today}
\author{
Hans van Ditmarsch
\institute{LORIA, CNRS \\ Universit\'e\ de Lorraine\\
       Nancy, France}
	\email{hans.van-ditmarsch@loria.fr}
	\and
Sophia Knight 
	\institute{LORIA, CNRS \\ Universit\'e\ de Lorraine\\
       Nancy, France}
       	\email{sophia.knight@gmail.com}
       \and
Ayb\"{u}ke \"{O}zg\"{u}n
\institute{LORIA, CNRS \\ Universit\'e\ de Lorraine\\
       Nancy, France}
	\email{aybuke.ozgun@loria.fr}
}       

\def\titlerunning{Announcement as effort on topological spaces}
\def\authorrunning{H. van Ditmarsch, S. Knight \& A. \"Ozg\"un}

\begin{document}

\maketitle

\begin{abstract}
We propose a multi-agent logic of knowledge, public and arbitrary announcements, that is interpreted on topological  spaces in the style of subset space semantics. The arbitrary announcement modality functions similarly to the effort modality in subset space logics,  however, it comes with intuitive and semantic differences. We provide axiomatizations for three logics based on this setting, and demonstrate their completeness.
\end{abstract}


\section{Introduction}\label{introduction}

In \cite{moss92}, Moss et al. introduce a bi-modal logic with language  $$\varphi::= p \ | \ \neg\varphi \ | \ \varphi\wedge\varphi \ | \  K\varphi \ | \ \Boxx\phi,$$ called subset space logic (SSL),  in order to formalize reasoning about sets and points together in one modal system. The main interest in their investigation lies in spatial structures such as topological spaces and using modal logic and the techniques behind for spatial reasoning, however,  they also have a strong motivation from epistemic logic.  While the modality $K$ is interpreted as knowledge, $\Boxx$ intends to capture the notion of \emph{effort}, i.e., any action that results in increase in knowledge. They propose subset space semantics for their logic. \ayycolor{A subset space is defined to be a pair $(X, \cO)$, where $X$ is a non-empty domain and $\cO$ is a collection of subsets of $X$  (not necessarily a topology),  wherein the modalities $K$ and $\Boxx$ are evaluated with respect to pairs of the form  $(x, U)$, where $x\in U\in\cO$.} \weg{\in\cO$ and $\cO$ is a collection of such subsets $U$ \cite{moss92}}According to subset space semantics, given a pair $(x, U)$, the modality $K$ quantifies over the elements of $U$,   whereas $\Boxx$ quantifies over all open subsets of $U$ that include the actual world $x$. Therefore, while knowledge is interpreted `locally' in a given observation set $U$, effort is read as \emph{open-set-shrinking} where more effort corresponds to a smaller neighbourhood, thus, a possible increase in knowledge. The schema $\Dia K\varphi$ states that after some effort the agent comes to know $\varphi$  where effort can be in the form of measurement, observation, computation, approximation \cite{moss92,parikh96,moss2,baskent12}, or announcement \cite{plaza:1989, balbiani08,eumas}.

The epistemic motivation behind the subset space semantics and the dynamic nature of the effort modality suggests a  link between SSL and dynamic epistemic logic, in particular dynamics known as public announcement \cite{baskent,baskent12,HvD-SSL,wang13,bjorndahl}. The works \cite{baskent,baskent12,HvD-SSL}  propose modelling public announcements on subset spaces by deleting the states or the neighbourhoods falsifying the announcement. This dynamic epistemic method is not in the spirit of the effort modality: dynamic epistemic actions result in global model change, whereas the effort modality results in local  neighbourhood shrinking. Hence, it is natural to search for an `open-set-shrinking-like' interpretation of public announcements on subset spaces. To best of our knowledge, Wang and {\AA}gotnes \cite{wang13} were the first to propose semantics for public announcements on subset spaces in the style of the effort modality, although this is not necessarily on topological spaces. Bjorndahl \cite{bjorndahl} then proposed a revised version of the \cite{wang13} semantics. In contrast to the aforementioned proposals, Bjorndahl uses models based on topological spaces to interpret knowledge and information change via public announcements. He considers the language $$\varphi::= p \ | \ \neg\varphi \ | \ \varphi\wedge\varphi \ | \  K\varphi \ | \ \intr{(\varphi)} \ | \ [\phi]\phi,$$ where $\intr{(\varphi)}$ means `$\varphi$ is true and can be announced', and where $[\phi]\psi$ means `after public announcement of $\phi$, $\psi$.'

In \cite{balbiani08}, Balbiani et al.\ introduce a logic to quantify over announcements in the setting of epistemic logic based on the language (with single-agent version here) $$\varphi::= p \ | \ \neg\varphi \ | \ \varphi\wedge\varphi \ | \  K \varphi \ | \ [\phi]\phi \ | \ \Box\varphi.$$  In this case, unlike above, $\Box\varphi$ means `after any announcement, $\varphi$ (is true)' so that $\Box$ quantifies over epistemically definable subsets ($\Box$-free formulas of the language) of a given model. In this case, $\Diamond K \varphi$ again means that the agent comes to know $\phi$, but in the interpretation that there is a formula $\psi$ such that after announcing it the agent knows $\varphi$. What becomes true or known by an agent  after an announcement can be expressed in this language without explicit reference to the announced formula. 

Clearly, the meaning of the effort $\Boxx$ modality and of the arbitrary announcement $\Box$ modality are related in motivation. In both cases, interpreting the modality requires quantification over sets. Subset-space-like semantics provides natural tools for this. In \cite{eumas}, we extended Bjorndahl's proposal \cite{bjorndahl} with an arbitrary announcement modality $$\varphi::= p \ | \ \neg\varphi \ | \ \varphi\wedge\varphi \ | \  K\varphi \ | \ \intr{(\varphi)} \ | \ [\phi]\phi \ | \ \Box\varphi$$ and provided topological semantics for the $\Box$ modality, and proved completeness for the corresponding single-agent logic $APAL_\intr{}$. 

In the current proposal we generalize this approach to a multi-agent setting. Multi-agent subset space logics have been investigated in \cite{heinemann08,heinemann10,baskent,agotnes13}.  There are some challenges with such a logic concerning the evaluation of higher-order knowledge. The general setup is for any finite number of agents, but to demonstrate the challenges, consider the case of two agents. Suppose for each of two agents $i$ and $j$ there is an open set such that the semantic primitive becomes a triple $(x,U_i,U_j)$ instead of a pair $(x,U)$. Now consider a formula like $K_i \M_j K_i p$, for `agent $i$ knows that agent $j$ considers possible that agent $i$ knows proposition $p$'. If this is true for a triple $(x,U_i,U_j)$, then $\M_j K_i p$ must be true for any $y \in U_i$; but $y$ may not be in $U_j$, in which case $(y,U_i,U_j)$ is not well-defined: we cannot interpret $\M_j K_i p$.  Our solution to this dilemma is to consider neighbourhoods that are not only relative to each agent, as usual in multi-agent subset space logics, but that are also {\em relative to each state}. This amounts to, when shifting the viewpoint from $x$ to $y \in U_i$, in $(x,U_i,U_j)$, we simultaneously have to shift the {\em neighbourhood} (and not merely the point in the actual neighbourhood) for the other agent. So we then go from $(x,U_i,U_j)$ to $(y,U_i,V_j)$, where $V_j$ may be different from $U_j$. If they are different, their intersection should be empty.

In order to define the evaluation neighbourhood for each agent with respect to the state in question, we employ a technique inspired by the standard neighbourhood semantics \cite{Chellas}.\weg{: we do not only use a set of neighbourhoods as one of the components of a topological model, as in \cite{moss2, HvD-SSL, bjorndahl,wang13} for the single-agent case, but} We use a set of \emph{neighbourhood functions}, determining the evaluation neighbourhood relative to both the given state and the corresponding agent. These functions need to be partial in order to render the  semantics well-defined for the dynamic modalities in the system.

\medskip

In Section \ref{sec.logic} we define the syntax, structures, and semantics of our multi-agent logic of arbitrary public announcements, $APAL_\intr{}$, interpreted on topological spaces equipped with a set of neighbourhood functions.  
 Without arbitrary announcements we get the logic $PAL_{int}$, and with neither arbitrary nor public announcements, the logic $EL_{int}$. In this section we also show some typical validities of the logic, and give a detailed example. In Section \ref{sec.axiomatization} we give axiomatizations for the logics: \ayycolor{$PAL_{int}$ extends $EL_{int}$ and $APAL_\intr{}$ extends $PAL_{int}$.} In Section \ref{sec.completeness} we demonstrate completeness for these logics. The completeness proof for the epistemic version of the logic, $EL_{int}$, is rather different from the completeness proof for the full logic $APAL_{int}$. We then compare our work to that of others (Section \ref{sec.comparison}) and conclude.

\section{The logic $APAL_\intr{}$} \label{sec.logic}

We define the syntax, structures, and semantics of our logic. From now on, $\mathit{Prop}$ is a countable set of propositional variables and $\cA$ a finite and non-empty set of agents.

\subsection{Syntax}

\begin{definition} 
The language $\mathcal{L}_{APAL_\intr{}}$ is defined by $$\varphi::= p \ | \ \neg\varphi \ | \ \varphi\wedge\varphi \ | \  K_i\varphi \ | \ \intr{(\varphi)} \ | \ [\varphi]\varphi \ | \ \Box\varphi$$  where $p\in \mathit{Prop}$ and $i\in \cA$.  Abbreviations for the connectives $\vee$, $\rightarrow$ and $\leftrightarrow$ are standard, and   $\bot$ is defined as abbreviation by $p\wedge\neg p$.  We employ $\M_i$ for $\neg K_i \neg \varphi$, and $\Diamond\varphi$ for $\neg\Box \neg\varphi$. We denote the non-modal part of  $\mathcal{L}_{APAL_\intr{}}$  (without the modalities $K_i$, $\intr{}$, $[\varphi]$ and $\Box$) by $\cL_{Pl}$, the part without $\Box$ by $\cL_{PAL_\intr{}}$, and the part without $\Box$ and $[\phi]$ by $\cL_{EL_\intr{}}$.
\end{definition}
Necessity forms \cite{goldblatt:1982} allow us to select unique occurrences of a subformula in a given formula (unlike in uniform substitution). They will be used in the axiomatization (Section \ref{sec.axiomatization}).
\begin{definition} 
Let $\varphi\in\cL_{APAL_\intr{}}$. The {\em necessity forms}  are inductively  defined as 
$$\xi(\sharp):= \sharp \ | \ \varphi\rightarrow \xi(\sharp) \ | \ K_i\xi(\sharp) \ | \ \intr{(\xi(\sharp))} \ | \ [\varphi]\xi(\sharp).$$
\end{definition}
It is not hard to see that each necessity form $\xi(\sharp)$ has a unique occurrence of $\sharp$.  Given a necessity form $\xi(\sharp)$ and a formula $\varphi\in\cL_{APAL_\intr{}}$, the formula obtained by replacing $\sharp$ by $\varphi$ is denoted by $\xi(\varphi)$.
\myComment{In the completeness proof (Section \ref{sec.completeness}) we use a complexity measure on formulas based on a notion $S$ called the formula's {\em size}. This  is a weighted count of the number of symbols of a formula. Another measure used is {\em $\Box$-depth}. This counts the number of the $\Box$-modalities occurring in a formula. The measure was introduced in \cite{HvD-simpleapal}. We let $\Sub(\varphi)$ denote the set of subformulas of a given formula $\varphi$.}

In the completeness proof (Section \ref{sec.completeness}) we use a complexity measure on formulas based on the \emph{size} and  {\em $\Box$-depth} of formulas where the size of a formula  is a weighted count of the number of symbols  and  {\em $\Box$-depth} counts the number of the $\Box$-modalities occurring in a formula. The measure was first introduced in \cite{HvD-simpleapal}.

\begin{definition}
\label{size}
The size $S(\varphi)$ of a formula $\varphi\in \cL_{APAL_\intr{}}$ is defined as: $S(p)=1$, $S(\neg\varphi)=S(\varphi)+1$, $S(\varphi\wedge\psi)=S(\varphi)+S(\psi)$, $S(K_i\varphi)=S(\varphi)+1$, $S(\intr{(\varphi)})=S(\varphi)+1$, $S([\varphi]\psi)=S(\varphi)+4S(\psi)$, and $S(\Box\varphi)=S(\varphi)+1$.
\end{definition}
The factor $4$ in the clause for $[\varphi]\psi$ is to ensure Lemma \ref{lemma2}. \ayycolor{Although the choice of the number $4$ might seem arbitrary, it is the smallest natural number guaranteeing the desired result (see the proof of Lemma \ref{lemma2})}.  
\begin{definition}
\label{depth}
The $\Box$-depth of a formula $\varphi\in \cL_{APAL_\intr{}}$, denoted by $d(\varphi)$, is defined as: $\ds(p)=0$,  $\ds(\neg\varphi)=\ds(\varphi)$, $\ds(\varphi\wedge\psi)=max \{\ds(\varphi), \ds(\psi)\}$, $\ds(K_i\varphi)=\ds(\varphi)$, $\ds(\intr{(\varphi)})=\ds(\varphi)$, $\ds([\varphi]\psi)=max \{\ds(\varphi), \ds(\psi)\}$, and $\ds(\Box\varphi)=\ds(\varphi)+1$.
\end{definition}
We now define three order relations on $\mathcal{L}_{APAL_\intr{}}$ based on the size and $\Box$-depth of the formulas.
\begin{definition}
\label{orders}
 For any $\varphi, \psi\in\cL_{APAL_\intr{}}$, 
\begin{itemize}  
\item $\varphi <^S \psi$ iff $S(\varphi)<S(\psi)$
\item $\varphi  <_d \psi$ iff $\ds(\varphi)<\ds(\psi)$
\item $\varphi\po\psi$ iff (either $\ds(\varphi)<\ds(\psi)$, or  $\ds(\varphi)=\ds(\psi)$  and $S(\varphi)< S(\psi)$)
\end{itemize}
\end{definition}
We let $\Sub(\varphi)$ denote the set of subformulas of a given formula $\varphi$.
\begin{lemma} \label{lemma1} \ayycolor{For any $\varphi, \psi \in \cL_{APAL_\intr{}}$,}
\begin{enumerate}  
\item $<^S,  <_d, \ <^S_d$ are well-founded strict partial orders between formulas in $\cL_{APAL_\intr{}}$,
\item $\varphi\in \Sub(\psi)$ implies $\varphi\po\psi$  \label{lemma1.2},
\ayycolor{\item $\intr{(\varphi)}\po [\varphi]\psi$, \label{lemma1.5}}
\item $\varphi\in\cL_{PAL_\intr{}}$ iff $\ds(\varphi)=0$, \label{lemma1.3} 
\item $\varphi\in\cL_{PAL_\intr{}}$ implies $[\varphi]\psi\po\Box \psi$. \label{lemma1.4}
\end{enumerate}
\end{lemma}

\begin{lemma}\label{lemma2}
For any $\varphi, \psi, \chi\in \cL_{APAL_\intr{}}$ and $i\in\cA$,
\begin{enumerate}  
\item $\neg [\varphi]\psi <^S_d [\varphi]\neg\psi$, \label{lemma2.1}
\ayycolor{\item $\intr{([\varphi]\psi)} <^S_d [\varphi]\intr{(\psi)}$, \label{lemma2.4}}
\item $K_i [\varphi]\psi <^S_d [\varphi] K_i\psi$,\label{lemma2.2}
\item $[\neg [\varphi]\neg\intr{(\psi)}]\chi <^S_d [\varphi][\psi]\chi$. \label{lemma2.3}
\end{enumerate}
\end{lemma}
\begin{proof}{} 
We only prove Lemma \ref{lemma2}.\ref{lemma2.3}. \ayycolor{The proof demonstrates why in the $[\varphi]\psi$ clause 
of Definition \ref{size}, 4 is the smallest natural number guaranteeing the result.} \aybuke{Is the previous sentence ok? I did not change it but it sounds a bit weird to me?}


By Definition  \ref{size}, we have  that $
S([\neg [\varphi]\neg\intr{(\psi)}]\chi )= S(\varphi)+ 4S(\psi)+4S(\chi)+9$ and that $S([\varphi][\psi]\chi)=S(\varphi)+ 4S(\psi)+16S(\chi)$. As for any $\chi\in\mathcal{L}_{APAL_\intr{}}$, $1 \leq S(\chi)$, it follows that $4S(\chi)+9 \leq 4S(\chi)+9S(\chi) = 13S(\chi) < 16S(\chi)$. Further, we observe that $d([\neg [\varphi]\neg\intr{(\psi)}]\chi) = \max \{d(\varphi), d(\psi), d(\chi)\} = d([\varphi][\psi]\chi)$. (This is similar in the first three items.) 
\end{proof}

\subsection{Background}
In this section, we introduce the topological concepts that will be used throughout this paper.  All the concepts in this section can be found in \cite{dugundji}.

\begin{definition}
A {\em topological space} $(X, \tau)$  is a pair  consisting of a non-empty set $X$ and a family $\tau$ of subsets of $X$ satisfying $\emptyset\in\tau$ and $X\in\tau$, and closed under finite intersections and arbitrary unions.
\end{definition}

The set $X$ is called the \emph{space}.\myComment{, the elements of $X$ are called \emph{points} of the space.} The subsets of $X$ belonging to $\tau$ are called \emph{open sets} (or \emph{opens}) in the space;  the family $\tau$ of open subsets of $X$ is also called a \emph{topology} on $X$. If for some $x\in X$ and an open $U\subseteq X$ we have $x\in U$, we say that $U$ is an \emph{open neighborhood} of $x$. 

A point $x$ is called an \emph{interior point} of a set $A\subseteq X$ if there is an open neighborhood $U$ of $x$ such that $U\subseteq A$. The set of all interior points of $A$ is called the \emph{interior} of $A$ and denoted by  $\Intr{(A)}$. We can then easily observe that for any $A\subseteq X$, $\Intr{(A)}$  is the largest open subset of $A$. 

\begin{definition}
A family $B\subseteq \tau$ is called a {\em base} for a topological space $(X, \tau)$ if every non-empty open subset of $X$ can be  written as a  union of elements of $B$.  
\end{definition}

\aycolor{Given any family $\Sigma=\{A_\alpha \ | \ \alpha\in I\}$ of subsets of $X$, there exists a unique, smallest topology $\tau(\Sigma)$ with $\Sigma\subseteq \tau(\Sigma)$ \cite[Th.\ 3.1]{dugundji}. The family $\tau(\Sigma)$ consists of $\emptyset$, $X$, all finite intersections of the  $A_\alpha$, and all arbitrary unions of these finite intersections. $\Sigma$ is called a \emph{subbase} for $\tau(\Sigma)$, and  $\tau(\Sigma)$ is said to be generated by $\Sigma$. The set of  finite intersections of members of $\Sigma$ forms a  base for $\tau(\Sigma)$.} 

\subsection{Structures}

In this section we define our multi-agent models based on topological spaces.

\ayycolor{\begin{definition}
\label{function}
Given a topological space $(X, \tau)$, a {\em neighbourhood function set} $\Phi$ on $(X, \tau)$ is a set of partial functions $\theta: X\pto\cA\imp\tau$ such that for all $x,y\in Dom(\theta)$, for all $i\in\cA$, and for all $U\in\tau$:
\begin{enumerate}  
\scomm{
\item \label{cond.1} $\theta(x)(i)\in\tau$,
\item \label{cond.2} $x\in \theta (x)(i)$,
\item \label{cond.4} $\theta(x)(i)\sse Dom(\theta)$,
\item \label{cond.5} if $y\in \theta (x)(i)$ then $\theta (x)(i)=\theta(y)(i)$,
\item \label{cond.6}$\theta|_U\in\Phi$, 
}
\end{enumerate}
where  $\theta|_U$ is the partial function with \mbox{$Dom(\theta|_U)=Dom(\theta)\cap U$} and $\theta|_U(x)(i)= \theta(x)(i)\cap U$. We call the elements of $\Phi$  {\em neighbourhood functions}.
\end{definition}}

\begin{definition}
\label{model}
A \emph{topological model with functions} (or in short, a \emph{topo-model}) is a tuple $\cM=\Model$, where $(X, \tau)$ is a topological space, $\Phi$ a neighbourhood function set, and \mbox{$V: Prop \imp X$} a valuation function. \ayycolor{We refer to the part ${\mathcal X}=(X, \tau, \Phi)$ without the valuation function as a {\em topo-frame}.}
\end{definition}

A pair $(x, \theta)$ is a {\em neighbourhood situation} if $x\in Dom(\theta)$ and $\theta(x)(i)$ is called the {\em epistemic neighbourhood at $x$ of agent $i$}. If $(x,\theta)$ is a neighbourhood situation in $\cM$ we write $(x,\theta)\in\cM$. Similarly, if $(x,\theta)$ is a neighbourhood situation in $\cX$ we write $(x,\theta)\in\cX$.

\begin{lemma}\label{dom-open}
For any $(X, \tau, \Phi)$ and $\theta\in\Phi$, $Dom(\theta)\in \tau$.
\end{lemma}
\weg{
\begin{proof}{}
Let $(X, \tau, \Phi)$ be a multi-agent subset space and $\theta\in \Phi$. We want to show that $Dom(\theta)\subseteq\Intr{Dom(\theta)}$. Let  $x\in Dom(\theta)$. Then, by Definition \ref{model}.\ref{cond.1},  \ref{model}.\ref{cond.2} and  \ref{model}.\ref{cond.4}, we have $x\in\theta(x)(i)\in\tau$ and $\theta(x)(i)\subseteq Dom(\theta)$. \myComment{ $\theta(x)(i)$ is an open neighbourhood of $x$ such that $\theta(x)(i)\subseteq Dom(\theta)$.}Therefore, by definition of the operator $\Intr$, $x\in \Intr{Dom(\theta)}$. Hence, $Dom(\theta)=\Intr{Dom(\theta)}$, i.e., $Dom(\theta)\in\tau$.
\end{proof}
}


\subsection{Semantics}\label{semantics}

\begin{definition}
Given a topo-model  \mbox{$\cM=\Model$}  and a neighbourhood situation $(x, \theta) \in \cM$, the semantics for the language $\mathcal{L}_{APAL_\intr{}}$ is defined recursively as: 
\[
\begin{array}{llll}
\mathcal{M}, (x, \theta) \models p &  \mbox{iff}& x\in V(p)\\
\mathcal{M}, (x, \theta) \models \neg\varphi & \mbox{iff}& \mbox{not} \ \mathcal{M}, (x, \theta)\models \varphi\\
\mathcal{M}, (x, \theta)\models \varphi\wedge\psi & \mbox{iff}& \mathcal{M}, (x, \theta)\models \varphi \ \mbox{and} \ \mathcal{M}, (x, \theta)\models \psi\\
\mathcal{M}, (x, \theta)\models K_i\varphi & \mbox{iff}& (\forall y \in \theta(x)(i))(\mathcal{M}, (y, \theta)\models \varphi)\\
\mathcal{M}, (x, \theta)\models \intr{(\varphi)} & \mbox{iff}& x\in\Intr{\br{\varphi}^{\theta}}\\
\cM,(x,\theta)\models[\varphi]\psi  & \mbox{iff} & \cM,(x,\theta)\models \intr(\varphi) \Imp\cM,(x,\theta^\varphi)\models\psi\\
\mathcal{M}, (x, \theta)\models \Box\varphi  & \mbox{iff}& (\forall \psi\in\cL_{PAL_\intr{}})(\mathcal{M}, (x, \theta)\models [\psi]\varphi)\\
\end{array}\] where $p\in\mathit{Prop}$, $\br{\varphi}^{\theta} =\{ y\in Dom(\theta) \ | \ \mathcal{M}, (y, \theta)\models\varphi\}$ and  \myComment{$\theta^\varphi:\Intr{\br{\varphi}^\theta}\imp \cA\imp\tau$} $\theta^\varphi:X\pto \cA\imp\tau$ such that $Dom(\theta^\varphi)= \Intr\br{\varphi}^\theta$ and  $\theta^\varphi(x)(i)=\theta(x)(i)\cap \Intr{\br{\varphi}^\theta}$.  
\end{definition}
The \emph{updated neighbourhood function} $\theta^\varphi$ is the restriction of $\theta$ to the open set $\Intr\br{\varphi}^\theta$, i.e., for all $x\in X$, $\theta^\varphi(x)(i)=\theta|_{\Intr\br{\varphi}^\theta}(x)(i)$. 

A formula $\varphi\in\mathcal{L}_{APAL_\intr{}}$ is \emph{valid in a topo-model} $\cM$, denoted $\cM\models\varphi$, iff $\cM, (x, \theta)\models\varphi$ for all $(x, \theta)\in \cM$; $\varphi$ is \emph{valid}, denoted $\models\varphi$, iff for all topo-models $\cM$ we have $\cM\models\varphi$. \ayycolor{Soundness and completeness with respect to topo-models are defined as usual.} 

Let us now elaborate on the structure of topo-models and the above semantics we have proposed for $\cL_{APAL_\intr{}}$. Given a topo-model $\Model$, the epistemic neighbourhoods of each agent at a given state $x$ are determined by  (partial) functions $\theta: X\pto\cA\imp\tau$ assigning an open neighbourhood to the state in question for each agent. We allow for partial functions in $\Phi$, and close $\Phi$  under taking restricted functions $\theta|_U$ where $U\in\tau$  (see Definition \ref{function}, condition \ref{cond.6}), so that updated neighbourhood functions are guaranteed to be well-defined elements of $\Phi$. As in the standard subset space semantics, by picking a neighbourhood situation $(x, \theta)$, we first  localize our focus to an \emph{open} subdomain, in fact to $Dom(\theta)$,  including the state $x$ and the epistemic neighbourhood  of each agent at $x$ determined by $\theta$. Then the function $\theta(x)$ designates an epistemic neighbourhood  for each agent $i$ in $\cA$.  It is guaranteed that every agent $i$ is assigned a neighbourhood by $\theta$ at every state $x$ in $Dom(\theta)$, since each $\theta(x)$ is defined to be a  \emph{total} function from $\cA$ to $\tau$. Moreover,  condition \ref{cond.2} of Definition \ref{function}  ensures that $\emptyset$ cannot be an epistemic neighbourhood, i.e., $\theta(x)(i)\not =\emptyset$ for all $x\in Dom(\theta)$. \ayycolor{Finally,  conditions \ref{cond.2} and \ref{cond.5} of Definition \ref{function} make sure that the $S5$ axioms for each $K_i$ are sound with respect to all topo-models.} 

\medskip

We now provide some semantic results. As usual in the subset space setting, truth of non-modal formulas only depends on the state in question. 
\begin{proposition}\label{eval.prop}
Give a topo-model $\cM=(X, \tau, \Phi, V)$, neighbourhood situations $(x, \theta_1), (x, \theta_2)\in \cM$, and a formula $\varphi\in \cL_{Pl}$. Then $(x, \theta_1)\models\varphi \mbox{ iff } (x, \theta_2)\models\varphi$.
\end{proposition}

\begin{proposition}\label{int=intint}
Given $\cM=(X, \tau, \Phi, V)$, $\theta \in\Phi$ and $\varphi\in \cL_{APAL_\intr{}}$. Then $\br{\intr{(\varphi)}}^\theta=\Intr{\br{\varphi}^\theta}$. 
\end{proposition}
\begin{proof}{}
\[
\begin{array}{llll}
\br{\intr{(\varphi)}}^\theta &  = & \{y\in Dom(\theta) \ | \ (y, \theta)\models\intr{(\varphi)}\} \\
 \  &  = & \{y\in Dom(\theta) \ | \ y\in\Intr{\br{\varphi}^\theta}\} \\
 \  &  = & \Intr{\br{\varphi}^\theta} \ \ \mbox{(since $\Intr{\br{\varphi}^\theta}\subseteq Dom(\theta)$)}\\
\end{array}\]\end{proof}

A corollary is that $\Intr{\br{\intr{(\varphi)}}^\theta}=\Intr\Intr{\br{\varphi}^\theta}=\Intr{\br{\varphi}^\theta}$.

\begin{proposition}\label{validities} \
\begin{enumerate}  
\item $\models [\varphi]\psi \leftrightarrow [\intr{(\varphi)}]\psi$\label{validities.1}
\item $\models (\intr{(\varphi)}\wedge \langle\varphi\rangle \intr{(\psi)})\leftrightarrow  \langle\varphi\rangle \intr{(\psi)}$\label{shortcut}
\end{enumerate}
\end{proposition}
\weg{
\begin{proof}{} 
We only show the first item.
\[
\begin{array}{lll}
 & (x, \theta)\models [\varphi]\psi \\ \mbox{iff} & (x,\theta)\models int(\varphi)\textrm{ implies } (x,\theta^\varphi)\models\psi \\
\mbox{iff} & (x,\theta)\models\intr{( int(\varphi))}\textrm{ implies } (x,\theta^\varphi)\models\psi \\ & \mbox{(by ($\intr{}$-T) and ($\intr{}$-4))} \\
\mbox{iff} & (x,\theta)\models\intr{( int(\varphi))}\textrm{ implies } (x,\theta^{\intr{(\varphi)}})\models\psi \\
\mbox{iff} & (x,\theta)\models  [\intr{(\varphi)}]\psi
\end{array}\]
\end{proof}
}

\begin{proposition}\label{update.fnc} \ 
\begin{enumerate}  
\item $\br{\psi}^{\theta^\varphi} = \br{\langle\varphi\rangle\psi}^\theta$\label{update.fnc1}
\item $\theta^\varphi = \theta^{\intr{(\varphi)}}$ \label{update.fnc2}
\item $(\theta^\varphi)^\psi = \theta^{\langle\varphi\rangle\intr{(\psi)}}$\label{update.fnc3}
\end{enumerate}
\end{proposition}

\subsection{Example}

We illustrate our logic by a multi-agent version of Bjorndahl's convincing example in  \cite{bjorndahl} about the jewel in the tomb. Indiana Jones ($i$) and Emile Belloq ($e$) are both scouring for a priceless jewel placed in a tomb. The tomb could either contain a jewel or not, the tomb could have been rediscovered in modern times or not, and (beyond \cite{bjorndahl}), the tomb could be in the Valley of Tombs in Egypt or not. The propositional variables corresponding to these propositions are, respectively, $j$, $d$, and $t$. We represent a valuation of these variables by a triple $xyz$, where $x,y,z \in \{0,1\}$. Given carrier set $X = \{ xyz \mid  x,y,z \in \{0,1\} \}$, the topology $\tau$ that we consider is generated by the base consisting of the subsets $\{000,100,001,101\}$, $\{010\}$, $\{110\}$, $\{011\}$, $\{111\}$. The idea is that one can only conceivably know (or learn) about the jewel or the location, on condition that the tomb has been discovered. Therefore, $\{000,100,001,101\}$ has no strict subsets besides empty set: if the tomb has not yet been discovered, no one can have any information about the jewel or the location.  

A topo-model $\mathcal{M} = (X,\tau,\Phi,V)$ for this topology $(X,\tau)$ has $\Phi$ as the set of all neighbourhood functions that are partitions of $X$ for both agents, and restrictions of these functions to open sets. A typical $\theta \in \Phi$ describes complete ignorance of both agents and is defined as $\theta(s)(i) = \theta(s)(e) = X$.  This corresponds most to the situation described in \cite{bjorndahl}. A more interesting neighbourhood situation in this model is one wherein Indiana and Emile have different knowledge. Let us assume that Emile has the advantage over Indiana so far, as he knows the location of the tomb but Indiana doesn't. This is the $\theta'$ such that for all $x \in X$, $\theta'(x)(i) =  X$ whereas the partition for Emile consists of sets $\{{101}, {100}, {001}, {000} \}$, $\{ {111}, {011} \}$, $\{ {110}, {010} \}$, i.e.,  $\theta'(111)(e) = \{ {111}, {011} \}$, etc. 

We now can evaluate what Emile knows about Indiana at $111$, and confirm that this goes beyond Emil's initial epistemic neighbourhood. This situation however does not create any problems in our setting since Indiana's epistemic neighbourhoods will be determined relative to the states in  Emile's initial neighbourhood. Firstly, Emile knows that the tomb is in the Valley of Tombs in Egypt 
\[ \mathcal{M}, ({111}, \theta') \models K_e t \] \weg{and therefore also that it has been discovered
\[ \mathcal{M}, ({111}, \theta') \models K_e d \] 
but} and he also knows that Indiana does not know that
\[ \mathcal{M}, ({111}, \theta') \models K_e \neg (K_i \neg t \vel K_i  t) \]
The latter involves verifying $\mathcal{M}, ({111}, \theta') \models \M_i  t$ and 
$\mathcal{M}, ({111}, \theta') \models \M_i \neg t$. And this is true because $\theta'(111)(i) = X$, and $000, 001 \in X$, and while $\mathcal{M}, (001, \theta') \models  t$, we also have $\mathcal{M}, (000, \theta') \models \neg t$. We can also check that Emile knows that Indiana considers it possible that Emile doesn't know the tomb's location
\[ \mathcal{M}, ({111}, \theta') \models K_e \M_i \neg (K_e t \vel K_e \neg t) \]
Announcements will change their knowledge in different ways. Consider the announcement of $j$.\ This results in Emile knowing everything but Indiana still being uncertain about the location.\weg{: `the jewel is in the tomb'. (Note that this announcement does not change the knowledge of the agents if the tomb has not been discovered. Its interior is then $\emptyset$.)}
\[ \mathcal{M}, ({111}, \theta') \models [j] (K_e (j\et d\et t) \et K_i (j\et d) \et \neg K_i (t \vel K_i \neg t)) \] Model checking this involves computing the epistemic neighbourhoods of both agents given by the updated neighbourhood function $(\theta')^j$ at $111$. Observe that $\Intr \II{j}^{\theta'}=\{111, 110\}$. Therefore, \\ $(\theta')^j(111)(e)=\Intr \II{j}^{\theta'}\cap\theta'(111)(e) = \{111\}$  and  $(\theta')^j(111)(i)=\Intr \II{j}^{\theta'}\cap \theta'(x)(i)=\{111, 110\}$.

\myComment{that $(\theta')^j(111)(i) = \Intr \II{j}^{\theta'} \inter \theta'(111)(i) = \Intr \II{j}^{\theta'} \inter X = \{ {111}, {110} \}$ and that $(\theta')^j(111)(e) = \Intr \II{j}^{\theta'} \inter \theta'(111)(e) = \{ {111}, {110} \} \inter \{ {111}, {011} \} = \{111\}$.}

There is an announcement after which Emile and Indiana know everything (for example the announcement of $j\et t$):
\[ \mathcal{M}, ({111}, \theta) \models \Diamond (K_e (j \et d \et t) \et K_i (j \et d \et t)) \]
As long as the tomb has not been discovered, nothing will make Emile (or Indiana) learn that it contains a jewel or where the tomb is located: \[ \mathcal{M} \models \neg d \rightarrow \Box (\neg (K_e j \vel K_e \neg j) \et \neg (K_e t \vel K_e \neg t))  \]

\section{Axiomatization} \label{sec.axiomatization}

We now provide the axiomatizations of $EL_\intr{}$, $PAL_\intr{}$, and $APAL_\intr{}$, and prove their soundness and completeness with respect to the proposed semantics.

\begin{table}[h!]
\begin{itemize}  
\item [] (P) \ \ \ \ \ \ \ \ \ \ all instantiations of propositional tautologies
\item [] ($K$-K) \ \ \ \ \ $K_i(\varphi\imp\psi)\imp (K_i\varphi \imp K_i\psi)$
\item [] ($K$-T) \ \ \ \ \ $K_i\varphi\imp \varphi$ 
\item [] ($K$-4) \ \ \ \ \ $K_i\varphi\imp K_iK_i\varphi$
\item [] ($K$-5) \ \ \ \ \ $\neg K_i\varphi\imp K_i\neg K_i\neg\varphi$
\item [] ($\intr{}$-K) \ \ \ $\intr{(\varphi\imp\psi)}\imp (\intr{(\varphi)} \imp \intr{(\psi)})$
\item [] ($\intr{}$-T) \ \ \ $\intr{(\varphi)}\imp \varphi$
\item [] $(\intr{}$-4) \ \ \ $\intr{(\varphi)}\imp \intr{(\intr{(\varphi)})}$
\item [] ($K_\intr{}$) \ \ \ \ \ $K_i\varphi\imp\intr{(\varphi)}$
\myComment{\item [(A2)] $[\chi](\varphi\imp\psi)\imp ([\chi]\varphi \imp [\chi]\psi)$
\item [] (R1) \ $\Box(\varphi\imp\psi)\imp (\Box\varphi \imp \Box\psi)$}
\item [] (R1) \ \ \ \ \ \ $[\varphi]p\leftrightarrow (\intr{(\varphi)}\imp p)$
\item [] (R2) \ \ \ \ \ \ $[\varphi]\neg\psi\leftrightarrow (\intr{(\varphi)}\imp\neg[\varphi]\psi)$
\item [] (R3) \ \ \ \ \ \ $[\varphi](\psi\wedge\chi)\leftrightarrow [\varphi]\psi\wedge[\varphi]\chi$
\ayycolor{\item [] (R4) \ \ \ \ \ \ $[\varphi]\intr{(\psi)}\leftrightarrow (\intr{(\varphi)}\imp \intr{([\varphi]\psi)})$}
\item [] (R5) \ \ \ \ \ \ $[\varphi]K_i\psi\leftrightarrow (\intr{(\varphi)}\imp K_i[\varphi]\psi)$
\item [] (R6) \ \ \ \ \ \ $[\varphi][\psi]\chi\leftrightarrow [\neg [\varphi] \neg\intr{(\psi)}]\chi$
\item [] (R7) \ \ \ \ \ \ $\Box\varphi\imp [\chi]\varphi$  \hspace{2cm} where $\chi\in \cL_{PAL_\intr{}}$
\item [] (DR1) \ \ \ From $\phi$ and $\phi \rightarrow \psi$, infer $\psi$ 
\item [] (DR2) \ \ \ From $\varphi$, infer $K_i\varphi$
\item [] (DR3) \ \ \ From $\varphi$, infer $\intr{(\varphi)}$
\item [] (DR4) \ \ \ From $\varphi$, infer $[\psi]\varphi$
\item [] (DR5) \ \ \ From $\xi([\psi]\chi)$ for all $\psi\in\mathcal{L}_{PAL_\intr{}}$, infer $\xi(\Box\chi)$
\end{itemize}
\caption{Axiomatizations $EL_\intr{}$, $PAL_\intr{}$, and $APAL_\intr{}$}
\label{axiomstable}
\end{table}
\begin{definition}
The axiomatization $APAL_\intr{}$ is given in Table \ref{axiomstable}. The axiomatization $PAL_\intr{}$ is the one without (DR5) and (R7). We get $EL_\intr{}$ if we further remove axioms (R1)-(R6) and the rule (DR4).
\end{definition}
The parts (DR1) to (DR5) are the {\em derivation rules} and the other parts are the {\em axioms}. A formula is a \emph{theorem} of $APAL_\intr{}$, notation $\vdash \varphi$, if it belongs to the smallest set of formulas  containing the axioms and closed under the derivation rules. (Similarly for $EL_\intr{}$ and $PAL_\intr{}$.)

\begin{lemma}
Axiomatization $APAL_\intr{}$ satisfies substitution of equivalents. If $\vdash \phi\leftrightarrow\psi$, then $\vdash \chi[p/\phi] \leftrightarrow \chi[p/\psi]$.
\end{lemma}
\begin{proof}{} \
In the above, $\chi[p/\phi]$ means uniform substitution of $\phi$ for $p$. The proof is not trivial but proceeds along similar lines as for public announcement logic, see \cite{hvdetal.del:2007}.
\end{proof}
\begin{proposition}\label{derivation}
$[\varphi]\bot\leftrightarrow \neg\intr(\varphi)$ is a theorem of $APAL_\intr{}$.
\end{proposition}
\weg{
\begin{proof}{}
Propositional steps are implicit.
\[
\begin{array}{lll}
1. & [\varphi](p\wedge \neg p) & \mbox{assumption} \\
2. & [\varphi]p\wedge[\varphi] \neg p & \mbox{(R3)}, 1. \\
3. & [\varphi]p & 2. \\
4. & \intr{(\varphi)} \imp p  & \mbox{(R1)}, 3. \\
5. & [\varphi] \neg p & 2. \\
6. & \intr{(\varphi)}\imp \neg [\varphi]p & \mbox{(R2)}, 5. \\
7. & \intr{(\varphi)}\imp \neg (\intr{(\varphi)}\imp p) & \mbox{(R1)}, 6. \\
8. & \intr{(\varphi)}\imp \neg p & 7. \\
9. & \intr{(\varphi)}\imp (p \wedge \neg p) & 4.,8.  \\
10. & \neg\intr{(\varphi)} & 9. \\
\end{array}\]
@@proof environment problem remove later@@
\end{proof}
}

\begin{proposition}
$APAL_\intr{}$ is sound with respect to the class of all topo-models.
\end{proposition}
\begin{proof}{}
Let  $\cM=\Model$ be a topo-model, $(x, \theta)\in\cM$ and $\varphi, \psi, \chi\in \cL_{APAL_\intr{}}$. We show three cases.

($\mathbf{K_\intr{}}$) \ \ Suppose $(x, \theta)\models K_i\varphi$. This means, $(y, \theta)\models\varphi$ for all $y\in \theta(x)(i)$. Hence, $\theta(x)(i)\subseteq \br{\varphi}^\theta$.  By Definition \ref{function},  $\theta(x)(i)$ is an open neighbourhood of $x$,  therefore we have $x\in \Intr{\br{\varphi}^\theta}$,  i.e., $(x, \theta)\models \intr{(\varphi)}$.

{\bf (R7)} \ \ Let  $\chi\in \cL_{PAL_\intr{}}$ and suppose $(x, \theta)\models \Box\varphi$. By the semantics, we have 
$ (x, \theta)\models \Box\varphi \mbox{ iff } (\forall \psi\in\cL_{PAL_\intr{}})( (x, \theta)\models [\psi]\varphi).$
Therefore, in particular, $(x, \theta)\models [\chi]\varphi$.

{\bf (DR5)} \ \  Suppose  $\xi([\psi]\chi)$ is valid for all $\psi\in\mathcal{L}_{PAL_\intr{}}$. The proof follows by induction on the complexity of $\xi(\sharp)$.  In case $\xi(\sharp)=\sharp$, we have $\xi([\psi]\chi)= [\psi]\chi$. By assumption, we have that $ [\psi]\chi$ is valid for all $\psi\in\mathcal{L}_{PAL_\intr{}}$. This implies $\cM, (x, \theta)\models [\psi]\chi$ for all $\psi\in\mathcal{L}_{PAL_\intr{}}$, all topo-models $\cM$, and $(x,\theta)\in\cM$. Therefore, by the semantics, $\cM,(x, \theta)\models \Box\chi$, i.e., $\cM,(x, \theta)\models \xi(\Box\chi)$. All other, inductive, cases are elementary.
\end{proof}
\begin{corollary}
The axiomatizations $EL_\intr{}$ and $PAL_\intr{}$ are sound with respect to the class of all topo-models.
\end{corollary}
\section{Completeness} \label{sec.completeness}

We now show completeness for $EL_\intr{}$, $PAL_\intr{}$, and $APAL_\intr{}$ with respect to the class of all topo-models. Completeness of $EL_\intr{}$ is shown in a standard way via a canonical model construction and a Truth Lemma that is proved by induction on formula complexity. 
Completeness for $PAL_\intr{}$ is shown by reducing each formula in $\cL_{PAL_\intr{}}$ to an equivalent formula of  $\cL_{EL_\intr{}}$. The proof of the completeness for $APAL_\intr{}$ becomes more involved. Reduction axioms for public announcements no longer suffice in the $APAL_\intr{}$ case, and the inductive proof needs a subinduction where announcements are \ayycolor{considered}. Moreover, the proof system of $APAL_\intr{}$ has an infinitary derivation rule, namely the rule (DR5), and given the requirement of closure under this rule, the maximally consistent sets for that case are defined to be maximally consistent {\em theories} (see, Section \ref{section.apal}). Lastly, the Truth Lemma requires the more complicated complexity measure on formulas defined in Section \ref{sec.logic}. There, we need to adapt the completeness proof of \cite{HvD-simpleapal} to our setting.

\subsection{Completeness of $EL_\intr{}$ and $PAL_\intr{}$}

For $\mathcal{L}_{EL_\intr{}}$ we define consistent and maximally consistent sets in the usual way, see e.g.\ \cite{bjorndahl} for details, and the multi-agent aspect does not complicate the definition. Let $X^c$ be the set of all maximally consistent sets of $EL_\intr{}$. We define relations $\sim_i$ on $X^c$ as $x\sim_i y \ \mbox{iff} \ \forall\varphi\in\mathcal{L}_{EL_\intr{}}(K_i\varphi\in x \ \mbox{iff} \ K_i\varphi\in y)$.  Notice that the latter is equivalent to: $\forall\varphi\in\mathcal{L}_{EL_\intr{}}(K_i\varphi\in x \ \mbox{implies} \ \varphi\in y)$ since $K_i$ is an $S5$ modality. As each $K_i$ is of $S5$ type, every $\sim_i$ is an equivalence relation, hence, it induces equivalence classes on $X^c$.  
Let $[x]_i$ denote the equivalence class of $x$ induced by the relation $\sim_i$. Moreover, we define $\widehat{\varphi}=\{y\in X^c \ | \ \varphi\in y\}$. Observe that $x\in\widehat{\varphi}$ iff $\varphi\in x$.

\myComment{\begin{lemma}\label{lemma11a}
For all $x$, $x+\varphi$ is consistent iff $\neg\varphi\not\in x$. 
\end{lemma}}

\begin{lemma}[Lindenbaum's Lemma]\label{lindenbaum2}
Each consistent\\ set can be extended to a maximally consistent set.
\end{lemma}

\myComment{\begin{lemma}\label{lemma5}
If $K_i\varphi\not\in x$, then there exists a maximally consistent set $y$ such that $K_ix\subseteq y$ and $\varphi\not \in y$.
\end{lemma}
\begin{proof}{}
Let $\varphi\in\mathcal{L}_{EL_\intr{}}$ and $x$ be such that $K_i\varphi\not\in x$. Thus, $\varphi \not\in K_ix$. Hence, by Lemma \ref{lemma11a}, $K_ix+\neg\varphi$ is consistent. Then, by Lemma \ref{lindenbaum2}, there exists a maximally consistent set $y$ such that $K_ix+\neg\varphi\subseteq y$. Therefore $K_i x\subseteq y$ and $\varphi\not\in y$.
\end{proof}}

\begin{definition}
\label{canon1}
We define the canonical model \\$\cX^c=(X^c, \tau^c, \Phi^c, V^c)$ as follows:
\begin{itemize}  
\item $X^c$ is the set of all maximally consistent sets;
\item $\tau^c$ is the topological space generated by the subbase $$\Sigma=\{[x]_i\cap \widehat{\intr{(\varphi)}} \ | \ x\in X^c, \varphi\in\mathcal{L}_{EL_\intr{}} \ \mbox{and} \ i\in \cA\};$$ 
\item $x\in V^c(p) \ \mbox{iff} \ p\in x, \ \mbox{for all} \ p\in\mathit{Prop}$;
\item $\Phi^c=\{\theta^*|_U\;|\;U\in\tau^c\}$, where  we define  $\theta^*:X^c\imp\cA\imp\tau^c$ as $\theta^*(x)(i)=[x]_i$, for $x\in X^c$ and $i\in\cA$.
\end{itemize}
\end{definition}

Observe that, since $\intrh{(\top)}=X^c$,  we have $[x]_i\cap \intrh{(\top)}=[x]_i\in\Sigma$ for each $i$. Therefore, each $[x]_i$ is an open subset of $X^c$. Moreover, the elements of $\Phi^c$ satisfy the required properties given in \mbox{Definition \ref{function}.}

\begin{lemma}[Truth Lemma] \label{zzz}
For every $\varphi\in\mathcal{L}_{EL_\intr{}}$ and for each $x\in X^c$, $\varphi\in x \ \mbox{iff} \ \cX^c, (x, \theta^*)\models \varphi.$
\end{lemma}

\begin{proof}{}
Cases for the propositional variables and Booleans are straightforward. We only show the cases for $K_i$ and $\intr{}$.

\textbf{Case} $\varphi:= K_i\psi$

($\Rightarrow$) \ \ Suppose $K_i\psi\in x$ and let $y\in\theta^*(x)(i)$. 
Since $y\in\theta^*(x)(i)=[x]_i$, by definition of $\sim_i$, we have $K_i\psi\in y$. Then, by T-axiom for $K_i$, we obtain $\psi\in y$. Then, by IH, $\cX^c, (y, \theta^*)\models \psi$. Therefore $\cX^c, (x, \theta^*)\models K_i\psi$.

($\Leftarrow$) \ \  Suppose $K_i\psi\not\in x$. Then, $\{K_i\gamma \ | \ K_i\gamma\in x\}\cup \{\neg\psi\}$ is a consistent set. We can then extend it to a maximally consistent set $y$.  As $\{K_i\gamma \ | \ K_i\gamma\in x\}\subseteq y$, we have $y\in [x]_i$ meaning that $y\in \theta^*(x)(i)$. Moreover,  since $\neg \psi\in y$, $\psi\not\in y$. Therefore, we have a maximally consistent set $y\in \theta^*(x)(i)$ such that $\psi\not\in y$. By (IH), $\cX^c, (y, \theta^*)\not\models \psi$. Hence, $\cX^c, (x, \theta^*)\not\models K_i\psi$.

\myComment{Suppose $K_i\psi\not\in x$. Then, by Lemma \ref{lemma5}, there exists a maximally consistent theory $y$ such that $K_ix\subseteq y$ and $\psi\not\in y$.  By IH, $(y, \theta^*)\not\models\psi$. Since $K_ix\subseteq y$,  we have $y\in [x]_i$ meaning that $y\in \theta^*(x)(i)$. Therefore, by the semantics, $\cX^c, (x, \theta^*)\not\models K_i\psi$.}

\textbf{Case} $\varphi:= \intr{(\psi)}$

($\Rightarrow$) \ \  Suppose $\intr{(\psi)}\in x$. Consider the set $[x]_i\cap \intrh{(\psi)}$ for some $i\in\cA$.  Obviously, $x\in[x]_i\cap \intrh{(\psi)}$ and $[x]_i\cap \intrh{(\psi)}$ is open (since it is in $\Sigma$). Now let $y\in[x]_i\cap \intrh{(\psi)}$. Since $y\in \intrh{(\psi)}$, $\intr{(\psi)}\in y$. Then, by ($\intr$-T), since $y$ is maximal consistent, we have $\psi\in y$. Thus, by IH, we have $(y, \theta^*)\models \psi$. Therefore, $y\in \br{\psi}^{\theta^*}$. This implies $[x]_i\cap \intrh{(\psi)}\subseteq\br{\psi}^{\theta^*}$. And, since $x\in[x]_i\cap \intrh{(\psi)}\in \tau^c$, we have $ x\in\Intr{\br{\psi}^{\theta^*}}$, i.e., $(x, \theta^*)\models \intr{(\psi)}$.

($\Leftarrow$)  \ \ Suppose  $(x, \theta^*)\models \intr{(\psi)}$, i.e., $ x\in\Intr{\br{\psi}^{\theta^*}}$.  Recall that the set of finite intersections of the elements of $\Sigma$ forms a base, which we denote by $B_\Sigma$, for $\tau^c$. 
$x\in\Intr{\br{\psi}^{\theta^*}}$ implies that there exists an open $U\in B_\Sigma$ such that $x\in U\subseteq \br{\psi}^{\theta^*}$.  Given the construction of $B_\Sigma$, $U$ is of the form 

$$U=\underset{i\in I_1} \bigcap [x_1]_i \cap \dots \underset{i\in I_n} \bigcap[x_k]_i \cap \underset{\eta\in \mathrm{Form_{fin}}}\bigcap \intrh{(\eta)}$$ where $I_1, \dots, I_n$ are finite subsets of $\cA$, $x_1\dots x_k\in X^c$ and $\mathrm{Form_{fin}}$ is a finite subset of $\cL_{EL_\intr{}}$. Since $\intr{}$ is a normal modality, we can simply write $$U=\underset{i\in I_1} \bigcap [x_1]_i \cap \dots \underset{i\in I_n} \bigcap[x_k]_i \cap \intrh{(\gamma)},$$ where $\underset{\eta\in \mathrm{Form_{fin}}}\bigwedge \eta:= \gamma$. \myComment{for some $\gamma\in\cL_{APAL_\intr{}}$.}Since   $x$ is in each $[x_j]_i$ with $1\leq j\leq k$, we have $[x_j]_i=[x]_i$ for all such $j$. Therefore, we have  $$x\in U= (\underset{i\in I}\bigcap[x]_i)   \cap \intrh{(\gamma)}\subseteq \br{\psi}^{\theta^*},$$ where $I= I_1\cup\dots \cup I_n$.


This implies, for all $y\in (\underset{i\in I}\bigcap[x]_i)$, if $y\in \intrh{(\gamma)}$ then $\psi\in y$.  From this, we can say $\underset{i\in I}\bigcup\{K_i\sigma \ | \ K_i\sigma \in x\}\vdash \intr{(\gamma)}\imp \psi$. 
Then, there is a finite subset $\Gamma\subseteq \underset{i\in I}\bigcup\{K_i\sigma \ | \ K_i\sigma \in x\}$ such  that $\vdash\underset{\lambda\in\Gamma}\bigwedge \lambda \imp (\intr{(\gamma)}\imp\psi)$.  It then follows: 
\[
\begin{array}{ll}
1. \vdash\intr{(\underset{\lambda\in\Gamma}\bigwedge \lambda \imp (\intr{(\gamma)}\imp\psi))} &   \mbox{(DR3)}\\
2. \vdash\intr{(\underset{\lambda\in\Gamma}\bigwedge \lambda)} \imp \intr{(\intr{(\gamma)}\imp\psi)}) &   \mbox{($\intr{}$-K)  and (DR1)}\\
3. \vdash(\underset{\lambda\in\Gamma}\bigwedge\intr{(\lambda)}) \imp \intr{(\intr{(\gamma)}\imp\psi)}) & \mbox{($\intr{}$-K)} 
\end{array}\]
Observe that each $\lambda\in \Gamma$ is of the form $K_j\alpha$ for some $K_j\alpha\in \underset{i\in I}\bigcup\{K_i\sigma \ | \ K_i\sigma \in x\}$ and we have $\vdash K_i\varphi\leftrightarrow\intr{(K_i\varphi)}$. Therefore, 
$\vdash(\underset{\lambda\in\Gamma}\bigwedge\lambda) \imp \intr{(\intr{(\gamma)}\imp\psi)})$.  Thus, since $\underset{\lambda\in\Gamma}\bigwedge\lambda\in x$ (by $\Gamma\subseteq x$),  we have $\intr{(\intr{(\gamma)}\imp\psi)})\in x$. Then, by ($\intr{}$-K), (DR1)   and since $\vdash\intr{(\intr{(\gamma)})}\leftrightarrow \intr{(\gamma)}$ and  $x\in\widehat{\intr{(\gamma)}}$ (i.e., $\intr{(\gamma)}\in x$) , we obtain $\intr{(\psi)}\in x$.
\end{proof}

\begin{theorem}\label{theorem2}
$EL_\intr{}$ is complete with respect to the class of all topo-models.
\end{theorem}



\begin{theorem}
$PAL_\intr{}$ is complete with respect to the class of all topo-models.
\end{theorem}
\begin{proof}{} 
\ayycolor{This follows from Theorem \ref{theorem2} by reduction in a standard way. The occurrences of the modality $\intr{}$ on the right-hand-side of the reduction axioms (axioms (R1)-(R6)) should not lead to any confusion: extending  the complexity measure defined in \cite[Definition 7.21 \ p. 187]{hvdetal.del:2007} to the language $\mathcal{L}_{PAL_\intr{}}$ by adding the same complexity measure for the modality $\intr{}$ as for $K_i$ gives us the desired result.} 
\end{proof}

\aybuke{I merged the seciton about completeness of $EL_\intr{}$ and $PAL_\intr{}$ because the latter was too short.  I added the necessary axiom to the axiomatization so we should not have a problem here. However, if necessary, we can add a complexity measure (following the one in the DEL book, p.187 and explain how reduction works but it will be like repetition. What do you think? I think the above explanation should be sufficient.}

\subsection{Completeness of $APAL_\intr{}$}\label{section.apal}

We now reuse the technique of \cite{HvD-simpleapal} in the setting of topological semantics. Given the closure requirement under derivation rule (DR5) it seems more proper to call maximally consistent sets \ayycolor{of $APAL_\intr{}$} maximally consistent theories, as further explained below.
\begin{definition}
A set $x$ of formulas is called a {\em theory} iff $APAL_\intr{}\subseteq x$ and $x$ is closed under (DR1) and (DR5). A theory $x$ is said to be consistent iff $\bot\not\in x$.  A theory $x$ is maximally consistent iff $x$ is consistent and any set of formulas properly containing $x$ is inconsistent.
\end{definition}
Observe that $APAL_\intr{}$ constitutes the smallest theory. Moreover, maximally consistent theories of $APAL_\intr{}$ posses the usual properties of maximally consistent sets:
\begin{proposition}\label{equivalent.defn}
For any maximally consistent theory $x$,  $\varphi\not\in x$ iff $\neg \varphi\in x$, and  $\varphi\land\psi\in x$ iff $\varphi\in x$ and $\psi\in x$.
\end{proposition}
\aybuke{Sophia changed this, please check. 
}

In the setting of our axiomatization based on the infinitary rule (DR5), we will say that a set $x$ of formulas is consistent iff there exists a consistent theory $y$ such that $x\subseteq y$.
Obviously, maximal consistent theories are maximal consistent sets of formulas. Under the given definition of consistency for sets of formulas, maximal consistent sets of formulas are also maximal consistent theories.
\begin{definition}\label{theory1}
Let $\varphi\in \cL_{APAL_\intr{}}$ and $i \in \cA$. Then $x+\varphi := \{\psi \ | \ \varphi\imp\psi\in x\}$\label{theory1.1} and $K_i x :=\{\varphi \ | \ K_i\varphi\in x\}$\label{theory1.2}.
\end{definition}
\begin{lemma} \label{lemma11z}
For any theory $x$ of $APAL_\intr{}$ and \\$\varphi\in \cL_{APAL_\intr{}}$, $x+\varphi$ is a theory and it contains $x$ and $\varphi$, and $K_ix$ is a theory. 
\end{lemma}
\begin{lemma}\label{lemma11a}
Let $\varphi\in\mathcal{L}_{APAL_\intr{}}$. For all theories $x$, $x+\varphi$ is consistent iff $\neg\varphi\not\in x$. 
\end{lemma}
\begin{proof}{}
Let $\varphi\in\mathcal{L}_{APAL_\intr{}}$ and $x$ be a theory. Then $\neg\varphi\in x$ iff $\varphi\imp\bot\in x$ (as $\neg\varphi \leftrightarrow \varphi\imp\bot$ is a theorem) iff $\bot\in x + \varphi$. Therefore, $x+\varphi$ is inconsistent iff $\neg\varphi\in x$, i.e., $x+\varphi$ is consistent iff $\neg\varphi\not\in x$.
\ayycolor{}
\end{proof}
\begin{lemma}[Lindenbaum's Lemma \cite{balbiani08}]\label{lindenbaum}
Each consistent theory can be extended to a maximal consistent theory.
\end{lemma}

\begin{lemma}\label{lemma5b}
If $K_i\varphi\not\in x$, then there is a maximally consistent theory $y$ such that $K_ix\subseteq y$ and $\varphi\not \in y$.
\end{lemma}

\begin{proof}{}
Let $\varphi\in\mathcal{L}_{APAL_\intr{}}$ and $x$ be such that $K_i\varphi\not\in x$. Thus, $\varphi \not\in K_ix$. Hence, by Lemma \ref{lemma11a}, $K_ix+\neg\varphi$ is consistent. Then, by Lemma \ref{lindenbaum}, there exists a maximally consistent set $y$ such that $K_ix+\neg\varphi\subseteq y$. Therefore $K_i x\subseteq y$ and $\varphi\not\in y$.
\end{proof}

\begin{lemma}\label{lemma6}
For all $\varphi\in \cL_{APAL_\intr{}}$ and all maximally consistent theories  $x$,  $\Box\varphi\in x$ iff  for all $\psi\in\cL_{PAL_\intr{}}$, \\ $[\psi]\varphi\in x$.
\end{lemma}
\begin{proof}{}
Let $\varphi\in \cL_{APAL_\intr{}}$ and $x$ be a maximally consistent theory.

($\Rightarrow$) \ \ Suppose $\Box\varphi\in x$. Then, by (R7) and (DR1), we have $[\psi]\varphi\in x$ for all $\psi\in\cL_{PAL_\intr{}}$.

($\Leftarrow$) \ \  Suppose  $[\psi]\varphi\in x$ for all $\psi\in\cL_{PAL_\intr{}}$. Consider the necessity form $\sharp$. By assumption, $\sharp ([\psi]\varphi)$  for all $\psi\in\cL_{PAL_\intr{}}$. Then, since $x$ is closed under (DR5), $\sharp (\Box\varphi)\in x$, i.e., $\Box\varphi\in x$ as well.
\end{proof}
The definition of  \emph{the canonical model for $APAL_\intr{}$} is the same as for $EL_\intr{}$, except that the maximally consistent sets are maximally consistent theories.\myComment{\begin{definition}[Canonical Model] 
\label{whatever}
See Def.\ \ref{canon1}.
\end{definition}}
We now come to the Truth Lemma for the logic $APAL_\intr{}$. Here we use the complexity measure $\psi\po\varphi$.
\myComment{\begin{definition}
Let $\varphi\in\cL_{APAL_\intr{}}$. Condition $P(\varphi)$ is defined as: ``For all maximally consistent theories $x$, $\varphi\in x$ iff $(x, \theta^*)\models \varphi$,'' and condition $H(\varphi)$ is defined as: ``For all formulas $\psi\in\cL_{APAL_\intr{}}$, if $\psi\po\varphi$, then $P(\psi)$.''
\end{definition}
\begin{lemma}\label{lemma7}
For all $\varphi\in\cL_{APAL_\intr{}}$, if $H(\varphi)$ then $P(\varphi)$.
\end{lemma}}

\begin{lemma}[Truth Lemma] \label{xxx}
For every $\varphi\in\mathcal{L}_{APAL_\intr{}}$ and for each $x\in X^c$, $\varphi\in x \ \mbox{iff} \ \cX^c, (x, \theta^*)\models \varphi$. 
\end{lemma}
\begin{proof}{} 
\myComment{Let $\varphi\in\mathcal{L}_{APAL_\intr{}}$ and $x\in \cX^c$. We prove that for all $\psi\in\cL_{APAL_\intr{}}$ such that $\psi\po\varphi$: $\psi\in x \ \mbox{iff} \ \cX^c, (x, \theta^*)\models \psi$, then $\varphi\in x \ \mbox{iff} \ \cX^c, (x, \theta^*)\models \varphi$. The proof is by induction on the structure of $\phi$, where the case $\phi = [\psi]\chi$ is proved by a subinduction on $\chi$. As $\po$ is a well-founded order, the proposition formulated in Lemma \ref{xxx} then follows directly.}

Let $\varphi\in\mathcal{L}_{APAL_\intr{}}$ and $x\in \cX^c$.  The proof is by $\po$-induction on $\varphi$, where the case $\phi = [\psi]\chi$ is proved by a subinduction on  $\chi$. We therefore consider 14 cases.

\textbf{Case} $\varphi:= p$ 
 \vspace{-0.2cm}
\[
\begin{array}{llllll}
x\in p & \mbox{iff} & x\in \nu^c(p)\nonumber \\
& \mbox{iff} & (x, \theta^*)\models p\nonumber
\end{array}\]

\noindent \textbf{Induction Hypothesis (IH)}: For all formulas $\psi\in\cL_{APAL_\intr{}}$, if $\psi\po\varphi$, then $\psi\in x \ \mbox{iff} \ \cX^c, (x, \theta^*)\models \psi$.

The cases negation, conjunction,  and interior modality are as in Truth Lemma \ref{zzz} for $EL_{int}$, where we observe that the subformula order is subsumed in the $<^S_d$ order (see Lemma \ref{lemma1}.\ref{lemma1.2}). We proceed with the knowledge operator, i.e., case $\varphi:=K_i\psi$, and \ayycolor{then with the subinduction on $\chi$ for case announcement $\varphi:= [\psi]\chi$, and finally with the case $\varphi:=\Box\psi$.}

\textbf{Case} $\varphi:= K_i\psi$ 

This case is also similar to the one in Truth Lemma \ref{zzz} for $EL_{int}$, however, using maximally consistent theories in the canonical model creates some differences. For the direction from left-to-right, see Truth Lemma \ref{zzz}. For ($\Leftarrow$), suppose $K_i\psi\not\in x$. Then, by Lemma \ref{lemma5b}, there exists a maximally consistent theory $y$ such that $K_ix\subseteq y$ and $\psi\not\in y$.  By  $\psi\po K_i\psi$ and (IH), $(y, \theta^*)\not\models\psi$. Since $K_ix\subseteq y$,  we have $y\in [x]_i$ meaning that $y\in \theta^*(x)(i)$. Therefore, by the semantics, $\cX^c, (x, \theta^*)\not\models K_i\psi$.

\textbf{Case} $\varphi:= [\psi] p$ 
\[
\begin{array}{llllll}
[\psi]p\in x &  \mbox{iff}& \intr{(\psi)}\imp p \in x  &  \text{(R1)} \\
\ & \mbox{iff}& \intr{(\psi)}\not\in x \ \mbox{or} \  p \in x \ & \text{Prop. \ref{equivalent.defn}} \\
\ & \mbox{iff}& (x, \theta^*)\not\models\intr{(\psi)} \ \mbox{or} \ (x, \theta^*)\models p  & (*) \\
\ & \mbox{iff}& (x, \theta^*)\models  [\psi]p  & \mbox{(R1)} \\
\end{array}\] 
(*): By (IH), $\intr{(\psi)}\po [\psi]p$ and $p\po [\psi]p$ (Lemma \ref{lemma1}.\ref{lemma1.5} and Lemma \ref{lemma1}.\ref{lemma1.2}).

\textbf{Case} $\varphi:= [\psi] \neg \eta$ \ \ Use (R2) and (IH) and, by  Lemma \ref{lemma1}.\ref{lemma1.5} and Lemma \ref{lemma2}.\ref{lemma2.1},  $\intr{(\psi)}\po [\psi]\neg\eta$ and \mbox{$\neg [\psi] \eta \po [\psi]\neg \eta$.}

\textbf{Case} $\varphi:= [\psi] (\eta\wedge\sigma)$ \ \ Use (R3) and (IH), $[\psi]\eta\po [\psi] (\eta\wedge\sigma)$ and  \mbox{$[\psi]\sigma\po [\psi] (\eta\wedge\sigma)$.}

\ayycolor{\textbf{Case} $\varphi:= [\psi] \intr{(\eta)}$ \ \ Use (R4) and (IH) and,  by  Lemmas \ref{lemma1}.\ref{lemma1.5}, \ref{lemma2}.\ref{lemma2.4}, $\intr{(\psi)}\po[\psi] \intr{(\eta)}$ and   \mbox{$\intr{( [\psi]\eta)}\po[\psi] \intr{(\eta)}$.}

\textbf{Case} $\varphi:= [\psi] K_i\eta$ \ \ Use (R5) and (IH) and,  by  Lemmas \ref{lemma1}.\ref{lemma1.5}, \ref{lemma2}.\ref{lemma2.2}, $\intr{(\psi)}\po[\psi] K_i\eta$ and   $K_i [\psi]\eta\po[\psi] K_i\eta$.}


\textbf{Case} $\varphi:= [\psi] [\eta]\sigma$ \ \ Use (R6) and (IH) and, by Lemma \ref{lemma2}.\ref{lemma2.3},  \mbox{$[\neg [\psi]\neg\intr{(\eta)}]\sigma\po [\psi][\eta]\sigma$.}  

\textbf{Case} $\varphi:= [\psi] \Box \sigma$ \ \ For all $\eta\in\cL_{PAL_\intr{}}$, $[\psi][\eta]\sigma\po  [\psi] \Box \sigma$, as  $[\psi] \Box \sigma$ has one more $\Box$ than $[\psi][\eta]\sigma$. \myComment{(also see Corollary \ref{complexity}).} Therefore, it  suffices to show $ [\psi] \Box \sigma\in x \ \mbox{iff} \ \forall \eta\in\cL_{PAL_\intr{}}, [\psi][\eta]\sigma\in x.$

($\Leftarrow$) \ \ Consider the necessity form $[\psi]\sharp$ and assume that for all $\eta\in\cL_{PAL_\intr{}}$, $[\psi][\eta]\sigma\in x$, i.e., for all $\eta\in\cL_{PAL_\intr{}}$, $[\psi]\sharp ([\eta]\sigma)\in x$ . As $x$ is closed under (DR5), we obtain $ [\psi]\sharp (\Box\sigma)\in x$, i.e., $[\psi]\Box\sigma\in x$. 

($\Rightarrow$) \ \ Suppose $[\psi]\Box \sigma\in x$.
We have 
\[
\begin{array}{ll}
\vdash \Box \sigma\to[\eta]\sigma, \ \mbox{for all} \ \eta\in\cL_{PAL_\intr{}} &  \mbox{(R7)}\\
\vdash [\psi](\Box \sigma\to[\eta]\sigma) \ \mbox{for all} \ \eta\in\cL_{PAL_\intr{}} & \mbox{(DR4)}\\
\vdash [\psi]\Box \sigma\to[\psi][\eta]\sigma, \ \mbox{for all} \ \eta\in\cL_{PAL_\intr{}} &  \mbox{(DR1), (R1-R3)}
\end{array}\]
Therefore, for all $\eta\in\cL_{PAL_\intr{}}, [\psi][\eta]\sigma\in x$.  As $[\psi][\eta]\sigma\po  [\psi] \Box \sigma$ for all $\eta\in\cL_{PAL_\intr{}}$, by (IH), we have for all $\eta\in\cL_{PAL_\intr{}}, (x, \theta^*)\models [\psi][\eta]\sigma$. Then, by the semantics, we obtain (details omitted) that $(x, \theta^*)\models [\psi] \Box\sigma$.
\weg{
\begin{eqnarray}
(\forall\eta\in\cL_{PAL_\intr{}}) (x, \theta^*)\models [\psi][\eta]\sigma & \mbox{iff} & (\forall\eta\in\cL_{PAL_\intr{}})((x, \theta^*)\models \intr{(\psi)} \ \mbox{implies} \ (x, (\theta^*)^\psi)\models[\eta]\sigma) \nonumber\\
 & \mbox{iff} & (x, \theta^*)\models \intr{(\psi)} \ \mbox{implies} \ (\forall\eta\in\cL_{PAL_\intr{}})( (x, (\theta^*)^\psi)\models[\eta]\sigma) \nonumber\\
  & \mbox{iff} & (x, \theta^*)\models \intr{(\psi)} \ \mbox{implies} \ (x, (\theta^*)^\psi)\models \Box\sigma \nonumber\\
  & \mbox{iff} &   (x, \theta^*)\models [\psi] \Box\sigma \nonumber
\end{eqnarray}
}

\textbf{Case} $\varphi:= \Box \psi$ \ \ Again note that  for all $\eta\in\cL_{PAL_\intr{}}$, $[\eta]\psi\po \Box \psi$, as  $\Box \psi$ has one more $\Box$ than $[\eta]\psi$ (see Lemma \ref{lemma1}.\ref{lemma1.3} and  Lemma \ref{lemma1}.\ref{lemma1.4}).  \myComment{Therefore, since $H(\Box \psi)$, we have $(\forall \eta\in\cL_{PAL_\intr{}})( [\eta]\psi \in x \ \mbox{iff} \ (x, \theta^*)\models [\eta]\psi)$, i.e., 
\begin{equation}\label{eqn.1}
(\forall \eta\in\cL_{PAL_\intr{}}) ([\eta]\psi \in x) \ \mbox{iff} \ (\forall \eta\in\cL_{PAL_\intr{}}) (x, \theta^*)\models [\eta]\psi
\end{equation}}
 Therefore, we obtain
\[
\begin{array}{llllll}
\Box \psi \in x &  \mbox{iff}&(\forall \eta\in\cL_{PAL_\intr{}}) ([\eta]\psi \in x) & \mbox{Lemma \ref{lemma6}} \\
\ & \mbox{iff}&(\forall \eta\in\cL_{PAL_\intr{}}) (x, \theta^*)\models [\eta]\psi & \mbox{(IH)} \\
\ & \mbox{iff}& \ (x, \theta^*)\models \Box\psi  & \mbox{semantics} \\
\end{array}\] \end{proof}

\begin{theorem}
$APAL_\intr{}$ is complete with respect to the class of all topo-models.
\end{theorem}
\begin{proof}{} 
Let $\varphi\in\cL_{APAL_\intr{}}$ such that $\not\vdash \varphi$, i.e., $\varphi\not\in APAL_\intr{}$  (Recall that $APAL_\intr{}$ is the smallest theory). Then, by  Lemma \ref{lemma11a},  $APAL_\intr{}+\neg\varphi$ is a consistent theory and,  by Lemma \ref{lemma11z}, $\neg\varphi\in APAL_\intr{}+\neg\varphi$. By Lemma \ref{lindenbaum}, the consistent theory $APAL_\intr{}+\neg\varphi$ can be extended to a maximally consistent theory $y$ such that $APAL_\intr{}+\neg\varphi\subseteq y$. Since $y$ is maximally consistent and $\neg\varphi\in y$, we obtain $\varphi\not\in y$ (by Proposition \ref{equivalent.defn}). Then, by Lemma \ref{xxx} (Truth Lemma), $\cX^c, (y, \theta^*)\not\models \varphi$.
\end{proof}

\section{Comparison to other work} \label{sec.comparison}

Multi-agent epistemic systems with subset space-like semantics have been proposed in \cite{heinemann08,heinemann10,baskent,agotnes13}, however, none of these are concerned with arbitrary announcements.  Our goal in this paper is not to provide a multi-agent generalization of SSL \emph{per se}, but to work with the \emph{effort-like} modality $\Box$ intended to capture the information change brought about by any announcements (subject to some restrictions) in a multi-agent setting and modelling it by way of ``open-set shrinking''  similar to the effort modality, rather than by deleting states or neighbourhoods, so that the intuitive link between the two becomes more transparent on a semantic level. In \cite{HvD-SSL}, Balbiani et al. proposed subset space semantics for arbitrary announcements, however, their approach does not go beyond the single-agent case and the semantics provided is in terms of model restriction.   An unorthodox approach to multi-agent knowledge is proposed in \cite{heinemann08,heinemann10}. Roughly speaking, instead of having a knowledge modality $K_i$ for each agent in his syntax, Heinemann uses additional operators to define $K_i$ and his semantics only validate the $S4$-axioms for $K_i$. The necessitation rule for $K_i$ does not preserve validity under the proposed semantics \cite{heinemann08,heinemann10}. In \cite{agotnes13} a multi-agent semantics for knowledge is provided, but no announcements or further generalizations (unlike in their other, single-agent, work \cite{wang13}), and not in a topological setting. Their use of partitions for each agent instead of a single neighbourhood is compatible with our requirement that all neighbourhoods for a given agent be disjoint. A further difference from the existing literature is that we restrict our attention to topological spaces and prove our results by means of topological tools. 

We applied the new completeness proof for arbitrary public announcement logic of \cite{HvD-simpleapal} to a topological setting. The canonical modal construction is as in \cite{bjorndahl} with some multi-agent modifications. The modality $\intr{}$  in our system demands a different complexity measure in the Truth Lemma of the completeness proof than in \cite{HvD-simpleapal}.

\section{Conclusions}
We have proposed topological semantics for the multi-agent extensions of the public announcement logic of \cite{bjorndahl}, and further extended the logic with arbitrary announcements. We showed topological completeness of these logics.  Our work can be seen as a  step toward discovering the interplay between dynamic epistemic logic and topological reasoning. 

For further research, we envisage a {\bf finitary} axiomatization for $APAL_{\intr{}}$ wherein the infinitary derivation rule (DR5) is replaced by a finitary rule. The obvious derivation rule would derive something after {\em any} announcement if it can be derived after announcing a fresh variable \cite{balbiani08}. Under subset space semantics, it is unclear how to prove that this rule is sound. 

We are still investigating expressivity and (un)decidability. If the logic $APAL_{\intr{}}$ is undecidable, this would contrast nicely with the undecidability of arbitrary public announcement logic. Otherwise, there may be interesting decidable versions when restricting the class of models to particular topologies.

The logic $APAL_{\intr{}}$ is also axiomatizable on the class where the $K$ modalities have $S4$ properties, a result we have not reported in this paper for consistency of presentation. This class is of topological interest.

In our setup all agents have the same observational powers. If agents can have different observational powers, we can associate a topology with each agent and generalize the logic to an arbitrary {\em epistemic action} logic. 

Furthermore, we would like to explore the exact difference between the effort modality and the arbitrary announcement modality (in the single agent case, see \cite{eumas})  by constructing a topological model which distinguishes the two: a topological model might have more than epistemically definable opens with respect to the proposed semantics. 

\section*{Acknowledgements}\label{sec:Acknowledgements}
We thank Philippe Balbiani for various detailed suggestions over the past year on how to improve our single-agent and multi-agent results in subset space logic and topological logics. We have found him very supportive of our efforts. We also thank the TARK reviewers for their valuable comments. Hans van Ditmarsch is also affiliated to IMSc (Institute of Mathematical Sciences), Chennai, as research associate. We acknowledge support from European Research Council grant EPS 313360.


\end{document}